\documentclass[a4paper,11pt]{article}
\pdfoutput=1

\usepackage[left=2.5cm, right=2.5cm, top=2.5cm, bottom=3cm]{geometry}
\usepackage{xspace}
\usepackage[protrusion=true,expansion=true]{microtype}
\usepackage{caption}
\usepackage{graphicx}
\usepackage{bbm,amsfonts}
\usepackage[fleqn]{amsmath}
\usepackage{ifthen}
\usepackage{paralist}
\usepackage{cite}
\usepackage{amsthm}
\usepackage[hang,flushmargin]{footmisc} 

\newtheoremstyle{theorems}
  {\topsep}   % ABOVESPACE
  {\topsep}   % BELOWSPACE
  {\itshape}  % BODYFONT
  {0pt}       % INDENT (empty value is the same as 0pt)
  {\bfseries} % HEADFONT
  {.}         % HEADPUNCT
  {5pt plus 1pt minus 1pt} % HEADSPACE
  {\thmname{#1}\thmnumber{ #2}\thmnote{ (#3)}} % CUSTOM-HEAD-SPEC

\newtheoremstyle{definitions}
  {\topsep}   % ABOVESPACE
  {\topsep}   % BELOWSPACE
  {}  % BODYFONT
  {0pt}       % INDENT (empty value is the same as 0pt)
  {\bfseries} % HEADFONT
  {.}         % HEADPUNCT
  {5pt plus 1pt minus 1pt} % HEADSPACE
  {\thmname{#1}\thmnumber{ #2}\thmnote{ (#3)}} % CUSTOM-HEAD-SPEC

\theoremstyle{theorems}
\newtheorem{theorem}{Theorem}
\newtheorem{lemma}{Lemma}
\newtheorem{corollary}{Corollary}

\theoremstyle{definitions}
\newtheorem{definition}{Definition}
\newtheorem{claim}{Claim}

\newcommand{\conbound}{Azuma-Hoeffding inequality\xspace}
\newcommand{\ie}{i.\,e.\xspace}
\newcommand{\eg}{e.\,g.\xspace}
\newcommand{\nat}{\mathbbm{N}}

\newcommand{\real}{\mathbbm{R}}
\newcommand{\tape}{\phi}
\newcommand{\algA}{\ensuremath{\mathtt{A}}\xspace}
\newcommand{\algB}{\ensuremath{\mathtt{A''}}\xspace}% to not conflict with the constant B
\newcommand{\algC}{\ensuremath{\mathtt{C}}\xspace}
\newcommand{\algR}{\ensuremath{\mathtt{R}}\xspace}
\newcommand{\adv}{\ensuremath{\mathtt{Adv}}\xspace}
\newcommand{\algOpt}{\ensuremath{\texttt{\textsc{Opt}}}\xspace}
\newcommand{\prob}{\ensuremath{\mathrm{Pr}}\mathopen{}}
\newcommand{\ev}{\ensuremath{\mathbbm{E}}\mathopen{}}
\newcommand{\costsymb}{\ensuremath{\mathrm{Cost}}\xspace}
\newcommand{\cost}[2][\empty]{%
  \ifthenelse{\equal{#1}{\empty}}{%
    \ensuremath{\costsymb}(#2)\xspace%
  }{%
    \ensuremath{\costsymb}_{#1}(#2)\xspace%
  }%
}
\newcommand{\opt}[1]{\cost[#1]{\algOpt}}
\newcommand{\optval}[1]{\ensuremath{\mathrm{opt}_{#1}}\xspace}%
\newcommand{\jss}{\textsc{Jss}\xspace}
\newcommand{\reqx}{\ensuremath{\hat{x}}}
\newcommand{\ansy}{\ensuremath{\hat{y}}}
\newlength{\punctuationfootlength}

\newcommand{\punctuationfootnote}[2]{#2\settowidth{\punctuationfootlength}%
  {#2}\hspace{-.3\punctuationfootlength}\footnote{#1}}

\newcommand{\emptyword}{\ensuremath{\lambda}\xspace}
\newcommand{\bW}{\overline{W}}

\newcommand{\ind}{\ensuremath{\overline{n}}}

\title{
    Randomized Online Computation with\\
    High Probability Guarantees\thanks{
        The research is partially funded by the SNF grant 200021--141089 and 
        Deutsche Forschungsgemeinschaft grant BL511/10-1.
    }
}

\newcommand{\email}[1]{{\tt #1}\xspace}
\newlength{\halftext}
\divide\halftext by 2
\advance\halftext by 0.9cm

\author{\begin{tabular}{cc}
          \parbox{\the\halftext}{\centering
               Dennis Komm\\
               Department of Computer Science\\
               ETH Zurich, Switzerland\\
               \email{dennis.komm@inf.ethz.ch}\phantom{\texttt{g}}\\
           }&
           \parbox{\the\halftext}{\centering
               Rastislav Kr\'alovi\v{c}\\
               Comenius University\\
               Bratislava, Slovakia\\
               \email{kralovic@dcs.fmph.uniba.sk}\\
           }\\
           \ & \\
           \parbox{\the\halftext}{\centering
               Richard Kr\'alovi\v{c}\\
               ETH Zurich / Google inc.\\
               Zurich, Switzerland\\
               \email{richard.kralovic@dcs.fmph.uniba.sk}
           }&
           \parbox{\the\halftext}{\centering
               Tobias M\"omke\\
               Department of Computer Science\\
               Saarland University, Germany\\
               \email{moemke@cs.uni-saarland.de}\phantom{\texttt{g}}
           }
           \end{tabular}
}
\date{}

\begin{document}

\maketitle

\begin{abstract}
  We study the relationship between the competitive ratio and the tail
  distribution of randomized online minimization problems.  To this end, we
  define a broad class of online problems that includes some of
  the well-studied problems like paging, $k$-server and metrical task systems on
  finite metrics, and show that for these problems it is possible to
  obtain, given an algorithm with constant expected competitive ratio, another
  algorithm that achieves the same solution quality up to an arbitrarily small
  constant error a with high probability; the ``high probability'' statement is
  in terms of the optimal cost.
  Furthermore, we show that our assumptions are tight in the sense that
  removing any of them allows for a counterexample to the theorem.
  In addition, there are examples of other problems not covered by our
  definition, where similar high probability results can be obtained.
\end{abstract}

\section{Introduction}\label{sec:intro}

In online computation, we face the challenge of designing algorithms that work
in environments where parts of the input are not known while parts of the output
(that may heavily depend on the yet unknown input pieces) are already needed.
The standard way of evaluating the quality of online algorithms is by means of
\emph{competitive analysis}, where one compares the outcome of an online
algorithm to the optimal solution constructed by a hypothetical optimal offline
algorithm.  Since deterministic strategies are often proven to fail for the most
prominent problems, randomization is used as a powerful tool to construct
high-quality algorithms that outperform their deterministic counterparts.  These
algorithms base their computations on the outcome of a random source; for a
detailed introduction to online problems we refer the reader to the
literature~\cite{BE98}.

The most common way to measure the performance of randomized algorithms is to
analyze the worst-case expected outcome and to compare it to the optimal
solution.  With offline algorithms, a statement about the expected outcome is
also a statement about the \emph{outcome with high probability} due to Markov's
inequality and the fact that the algorithm may be executed many times to
amplify the probability of success~\cite{Hro05}.  However, this amplification
is not possible in online settings.  As online algorithms only have one
attempt to compute a reasonably good result, a statement with respect to the
expected value of their competitive ratio may be rather unsatisfying.  As a
matter of fact, for a fixed input, it might be the case that such an
algorithm produces results of a very high quality in very few cases (\ie, for a
rather small number of random choices), but is unacceptably bad for the
majority of random computations; still, the expected competitive ratio might
suggest a better performance.  Thus, if we want to have a certain guarantee
that some randomized online algorithm obtains a particular quality, we must
have a closer look at its analysis.  In such a setting, we would like to state
that the algorithm does not only perform well on average, but ``almost
always.''

Besides a theoretical formalization of the above statement, the main
contribution of this paper is to show that, for a broad class of problems, the
existence of a randomized online algorithm that performs well in expectation
immediately implies the existence of a randomized online algorithm that is
virtually as good with high probability.  Our investigations, however, need to
be detailed in order to face the particularities of the framework. 
First, we show that it is not possible to measure the
probability of success with respect to the input size, which might be
considered the straightforward approach.  
Many of the known randomized online algorithms are naturally divided into some
kind of \emph{phases} (\eg, the algorithm for metrical task systems from
Borodin et al.~\cite{BLS92}, the marking algorithm for paging from Fiat et
al.~\cite{FKL+91}, etc.) where each phase is processed and analyzed separately.
Since the phases are independent, a high probability result (\ie, with a
probability converging to $1$ with an increasing number of phases) can be
obtained.  However, the definition of these phases is specific to each problem
and algorithm.  Also, there are other algorithms (\eg, the optimal paging
algorithm from Achlioptas et al.~\cite{ACN00} and many workfunction-based
algorithms) that use other constructions and that are not divided into phases.
As we want to establish results with high probability that are independent of
the concrete algorithms, we thus have to measure this probability with respect
to another parameter; we show that the cost of an optimal solution is a very
reasonable quantity for this purpose.

Then again it turns out that, if we consider general online problems, the
notions of the expected outcome and an outcome with high probability are still
not related in any way, \ie, we define problems for which these two measures
are incomparable.  Hence, we carefully examine both to which
parameter the probability should relate and which properties we need the
studied problem to fulfill to again allow a division into independent phases;
finally, this allows us to construct randomized online algorithms that perform
well with a probability tending to $1$ with a growing size of the optimal cost.
We show that this technique is applicable for a wide range of online problems.

Classically, results concerning randomized online algorithms commonly analyze
their expected behavior; there are, however, a few exceptions, \eg, Leonardi
et al.~\cite{LMPR01} analyze the tail distribution of algorithms for call
control problems, and Maggs et al.~\cite{MMVW97} deal with online distributed
data management strategies that minimize the congestion in certain network topologies.

\subsection*{Overview of this Paper}

In Section~\ref{sec:prelim}, we define the class of symmetric online
minimization problems and present the main result (Theorem~\ref{thm:exptohp}).
The theorem states that, for any symmetric problem which fulfills certain
natural conditions, it is possible to transform an algorithm with constant
expected competitive ratio $r$ to an algorithm having a competitive ratio of
$(1+\varepsilon)r$ with high probability (with respect to the cost of an
optimal solution).  Section~\ref{sec:mainthm} is devoted to proving
Theorem~\ref{thm:exptohp}.  We partition the run of the algorithm into phases
such that the loss incurred by the phase changes can be amortized;  however, to
control the variance within one phase, we need to further subdivide the phases.
Modelling the cost of single phases as dependent random variables, we obtain a
supermartingale that enables us to apply the \conbound and thus to
obtain the result.  These investigations are followed by applications of the
theorem in Section~\ref{sec:applications} where we show that our result is
applicable for task systems and that for the $k$-server problem on
unbounded metric spaces, no comparable result can be obtained.  We further
elaborate on the tightness of our result in Section~\ref{sec:discussion}.

\section{Preliminaries}\label{sec:prelim}

We use the following definitions of online algorithms \cite{BE98} that deal
with online minimization problems. 

\begin{definition}[Online Algorithm]\label{dfn:online-alg}
  Consider an initial configuration $I$ and an input sequence
  $x=(x_1,\dots,x_n)$.
  An \emph{online algorithm} \algA computes the output sequence
  $\algA(I,x)=(y_1,\dots,y_n)$, where  $y_i=f(I,x_1,\dots,x_i)$ for some function $f$.
  The \emph{cost} of the solution $\algA(I,x)$ is denoted by $\cost[I,x]{\algA}$.
\end{definition}

For the ease of presentation, we refer to the tuple that consists of the
initial configuration and the input sequence, \ie, $(I,x)$, as the input of the problem.
Even though the initial configuration is not explicitly introduced in the
definition in \cite{BE98}, it is often very natural, and it is used in the
definitions of some well-known online problems (\eg, the $k$-server problem
\cite{Kou09}).  As we see later, the notion of an initial configuration plays an
important role in the relationship between different variants of the
competitive ratio.

Since, for the majority of online problems, deterministic strategies are
often doomed to fail in terms of their output quality, randomization is
used in the design of online algorithms~\cite{Hro05,BE98,IK97}.  Formally,
randomized online algorithms can be defined as follows.

\begin{definition}[Randomized Online Algorithm]
  A \emph{randomized online algorithm} $\algR$ computes the output sequence
  $\algR^{\tape}(I,x)=(y_1,\dots,y_n)$
  such that $y_i$ is computed from $\tape,I,x_1,\dots,x_i$, where $\tape$
  is the content of the random tape, \ie, an infinite binary sequence where
  every bit is chosen uniformly at random and independent of the others.
  By $\cost[I,x]{\algR}$ we denote the random variable (over the
  probability space defined by $\tape$) expressing the cost
  of the solution  ${\algR}^{\tape}(I,x)$.
\end{definition}

The efficiency of an online algorithm is usually measured in terms of
the competitive ratio as introduced by Sleator and Tarjan~\cite{ST85}.

\begin{definition}[Competitive Ratio]
  An online algorithm is $r$-\emph{competitive}, for some $r\ge 1$, if there exists a constant
  $\alpha$ such that, for every initial configuration $I$ and each input
  sequence $x$, $\cost[I,x]{\algA}\le r\cdot \opt{I,x}+\alpha$,
  where $\opt{I,x}$ denotes the value of the optimal solution for the given instance;
  an online algorithm is \emph{optimal} if it is $1$-competitive with
  $\alpha=0$.
\end{definition}

When dealing with randomized online algorithms we compare the
expected outcome to the one of an optimal algorithm. 

\begin{definition}[Expected Competitive Ratio]
  A randomized online algorithm $\algR$ is $r$-com\-pet\-i\-tive\footnote{The notion 
  of competitiveness for randomized online algorithms as used in this paper is called
  competitiveness against an \emph{oblivious adversary} in the literature. For an
  overview of the different adversary models, see, \eg, \cite{BE98}.} \emph{in expectation} 
  if there exists a constant $\alpha$ such that, for every initial configuration $I$ and
  input  sequence $x$, $\ev\left[\cost[I,x]{\algR}\right] \le
  r\cdot\opt{I,x}+\alpha$.
\end{definition}

In the sequel, we analyze the notion of \emph{competitive ratio with high
probability}. Before stating the definition, however, we quickly discuss what
parameter the high probability should relate to.  As already mentioned, a natural way would be to
define an event to have high probability if the probability that it appears
tends to $1$ with increasing input length (\ie, the number of requests).
However, this does not seem to be very useful; consider, \eg, the well-known
paging problem~\cite{BE98,IK97} with cache size $k$ (we describe and study paging more thoroughly
in Subsection~\ref{subsec:paging}): For any input $x$ of length $n$ and any
competitive ratio $r$ and any $d$, there is an input $x'$ of length $dn$ formed by repeating
every request $d$ times.  Hence, for any algorithm, the performance on $x$
and $x'$ is the same\punctuationfootnote{This is true if we assume that the algorithm
only deletes a page from its buffer if a page fault occurs, which is implied by
the problem definition, see~\cite{IK97}.}.  This implies that there is no
randomized algorithm for paging that achieves a competitive ratio of less than
$k$ with a probability approaching $1$ with growing $n$.  Let $r<k$ and suppose
that there exists $n_0\in\nat$ and a randomized online algorithm \algR that,
for any input $x$ with $|x|=n\ge n_0$ is $r$-competitive with probability
$1-1/f(n)$, for some function $f$ that tends to infinity with growing $n$.
Thus, there is a randomized online algorithm $\algR'$ that is $r$-competitive on
every instance $x'$ independent of its length with this probability.  In
particular, if there exists such an algorithm, then there exists a randomized
online algorithm \algC that is $r$-competitive on instances of length $k$ with
probability $1-1/f(n)$, for any $n$.  Now consider the following instance that
consists of $k$ requests and let the
cache be initialized with pages $1,\dots,k$; an adversary requests page $k+1$
at the beginning and a unique page in the next $k-1$ time steps.  Clearly,
there exists an optimal solution with cost $1$.  In every time step in which a
page fault occurs, \algC, using its random source, chooses a page to evict to make space in the
cache.  Since the adversary knows \algC's probability distribution, without
loss of generality, we assume that \algC chooses every page with the same
probability.  Note that there exists a sequence $p_1,\dots,p_k$ of ``bad''
choices that causes \algC to have cost $k$.  In the first time step, \algC
chooses the bad page with probability at least $1/k$; with probability at least
$1/k^2$, it chooses the bad pages in the first and the second time step and so
on.  Clearly, the probability that it chooses the bad sequence is at least
$1/k^k$.  But this immediately contradicts that \algC performs well on this
instance with probability $1+1/f(n)$, for arbitrarily large~$n$.

Then again, for the practical use of paging algorithms, the instances where
also the optimal algorithm makes faults are of interest.  Hence, it seems
reasonable to define the term \emph{high probability} with respect to the cost of an
optimal solution.  In this paper, we use a strong notion of high probability
requiring the error probability to be subpolynomial.

\begin{definition}[Competitive Ratio w.h.p.]\label{dfn:compwhp}
  A randomized online algorithm \algR is $r$-com\-pet\-i\-tive \emph{with high
  probability} (w.h.p. for short) if, for any  $\beta\ge1$, there exists a
  constant $\alpha$ such that for all initial configurations and inputs $(I,x)$ it holds that
  \[ \prob\left[\cost[I,x]{\algR}\ge r\cdot
     \opt{I,x}+\alpha\right]\le\left(2+\opt{I,x}\right)^{-\beta}. \]
\end{definition}

First, note that the purpose of the constant $2$ on the right-hand side of the
formula is to properly handle inputs with a small (possibly zero) optimum.  The
choice of the particular constant is somewhat arbitrary (however, it should be
greater than $1$) since the $\alpha$ term on the left-hand side hides the
effects.  We now show that the two notions of the expected and the high-probability
competitiveness are incomparable. Let $[n]$ denote the set $\{1,\dots,n\}$.

\begin{enumerate}
  \item On the one hand, there are problems for which the competitive ratio w.h.p. is
    better than the expected one. Consider, \eg, the following problem. There
    is a unique initial configuration $I$ and the input sequence consists of
    $n+1$ bits $x_0=1,x_1,\dots,x_n$. An online algorithm has to produce
    one-bit answers $y_0,\dots,y_{n+1}$.
    If, for every $i\in [n]$,
    it holds that $y_{i-1}=x_i$, the cost is $D=2^{2n}$, otherwise the cost is
    $\sum_{i=0}^nx_i$, which is optimal. A straightforward algorithm that
    guesses each bit with probability $1/2$ has probability $1-1/2^n$
    to be optimal on every input. 
    
    Consider some $\beta\ge1$; let $n_\beta$ be the smallest integer such that
    $2^{n_\beta}\ge(n_\beta+3)^\beta$ and let $\alpha=2^{2n_\beta}$.  For any
    input of length $n\ge n_\beta$ we have
    \[ \prob\left[\cost[I,x]{\algR}>\opt{I,x}\right]\le\frac{1}{2^n} 
       \le\frac{1}{(n+3)^\beta}\le\frac{1}{(2+\opt{I,x})^\beta}. \]
    For inputs of length at most $n_\beta$, any solution has a cost of at most $\alpha$, so 
    \[ \prob\left[\cost[I,x]{\algR}>\alpha\right]=0. \]
    Hence, the algorithm is $1$-competitive w.h.p.
  
    However, for any algorithm, there is an input such that the probability of
    guessing the whole sequence is at least $1/2^n$, so the expected
    cost is at least $D/2^n$. Since the optimum is at most $n+1$, any algorithm
    has an expected competitive ratio of at least 
    \[ \frac{D}{(n+1)2^n}=\frac{2^n}{n+1}. \]
  \item On the other hand, the following problem shows that sometimes the expected
    performance is better than the one we get w.h.p.  In fact, we show that
    the gap between these two measures can be arbitrarily large. Consider a
    problem with $n$ requests, where the first $n-1$ ones are just dummy requests
    that serve for padding, and the last one is $x_n\in\{1,\dots,b\}$  for some
    positive integer $b$ that depends on $n$. The answer $y_1$ has to be a
    number from the interval $\{1,\dots,b\}$. The cost is $n$ if
    $y_1\not=x_n$, and $bn$ otherwise. An algorithm that chooses $y_1$
    uniformly at random pays $n$ with probability $(b-1)/b$ and $bn$
    with probability $1/b$; hence the expected cost is
    $n(1+(b-1)/b)\le 2n$.  However, there is always an input such that the
    probability to pay $bn$ is at least $1/b$.  For any $k$, we can choose
    $b:=n^k$.  Then, no algorithm can achieve a solution with cost better than $n^{k+1}$
    with probability at least $1-1/n^k$. Since the optimal cost is $n$, there
    is no algorithm with competitive ratio $n^k$ w.h.p., but
    there is one with an expected competitive ratio of $2$.
\end{enumerate}

However, the problems used in the previous examples were quite artificial; many
real-world online problems share additional properties that guarantee a
closer relationship between the expected and high-probability behavior.
In what follows, we thus focus on so-called \emph{partitionable} problems.

\begin{definition}[Partitionability]\label{dfn:partitionable}
  An online problem is called \emph{partitionable} if there is a non-negative
  function $\mathcal{P}$ such
  that, for any initial configuration $I$, the sequence of requests $x_1,\dots,x_n$, and  
  the corresponding solutions $y_1,\dots,y_n$, we have
  \[ \cost[I,x]{y_1,\dots,y_n}=\sum_{i=1}^n\mathcal{P}(I, x_1,\dots,x_i;y_1,\dots,y_i). \]
\end{definition}

In other words, for a partitionable problem, the cost of a solution is the sum
of the costs of particular answers, and the cost of each answer is independent
of the future input and output.  The partitionability allows us to speak of the
cost of a subsequence of the outputs.  A problem can only fail to be
partitionable if the cost may decrease with additional request-answer pairs.
We can, however, transform every online problem into a partitionable one by
introducing a dummy request at the end as a unique end marker.  This way, we can
assign a value of zero to all answers but the last one.  Therefore, the
partitionability condition stated in this way causes no restriction on the
online problem.  However, we further restrict the behavior, and it will be
convenient to think in terms of the ``cost of a particular answer.''

\begin{definition}[Request-Boundedness]\label{dfn:req-bounded}
  Let the function $\mathcal{P}$ be defined as in
  Definition~\ref{dfn:partitionable}.  A partitionable problem is called
  \emph{request-bounded} if, for some constant $F$, we have
  \[ \forall I,x,y,i\colon\mathcal{P}(I,x_1,\dots,x_i;y_1,\dots,y_i)\le F\quad\text{or}\quad
     \mathcal{P}(I,x_1,\dots,x_i;y_1,\dots,y_i)=\infty. \]
    %We need the \infty-part for Section 3.1 in order to make exclude expensive answers
\end{definition}

Note that for any partionable problem there is a natural
notion of a state; for instance, it is the content of the memory for the paging
problem, the position of the servers for the $k$-server problem, etc.  Now we
provide a general definition of this notion.  By $a\cdot b$, we denote
the concatenation of two sequences $a$ and $b$; \emptyword denotes the empty sequence. 

\begin{definition}[State]\label{dfn:state}
  Consider two initial configurations $I$ and $I'$, two sequences of
  requests $x=(x_1,\dots,x_n)$ and $x'=(x_1',\dots,x_m')$, and two
  sequences of outputs $y=(y_1,\dots,y_n)$ and $y'=(y'_1,\dots,y'_m)$.
  The triples $(I,x,y)$ and $(I',x',y')$ are equivalent if,
  for any sequence of requests $x''=(x''_1,\dots,x''_p)$ and a
  sequence of outputs $y''=(y_1'',\dots,y_p'')$,
  the input $(I,x\cdot x'')$ is valid with a solution $y\cdot y''$
  if and only if the input $(I',x'\cdot x'')$ is valid with a solution
  $y'\cdot y''$, and the cost of $y''$ is the same for both solutions.
 
  A \emph{state} $s$ of the problem is an equivalence class over the triples
  $(I,x,y)$.  Let $(I,x,y)$ be some triple in $s$. By
  $\algOpt_s(x')$ we denote the sequence of outputs $y'$ such that
  $y\cdot y'$ is a valid solution of the input $(I,x\cdot x')$ and the
  cost $\cost[J, x']{\algOpt_s(x')}$ of $y'$ is minimal, where $J$
  is the configuration determined by $(I,x,y)$.
  A state $s$ is an \emph{initial} state if and only if it contains some triple
  $(I,\emptyword,\emptyword)$.
\end{definition}

We chose this definition of states as it covers best the properties of
online computations as we need them in our main theorem.  An alternative
definition could use task systems with infinitely many states, but the
description would become less intuitive; we will return to task systems
in Section~\ref{sec:task}.

From now on we sometimes slightly abuse notation and write
$\cost{\algOpt_s(x')}$ instead of $\cost[J,x']{\algOpt_s(x')}$ if
the configuration $J$ corresponds to a triple in $s$, as it is sufficient to know
the state $s$ instead of $J$ in order to determine the value of the function. 
Intuitively, a state from Definition~\ref{dfn:state} encapsulates all
information about the ongoing computation of the algorithm that is relevant for
evaluating the efficiency of the future processing.  Usually, the state is
naturally described in the problem-specific domain (content of cache, current
position of servers, set of jobs accepted so far, etc.).  Note, however, that
the internal state of an algorithm is a different notion since it may, \eg,
behave differently if the starting request had some particular value. The
following properties are crucial for our approach to probability amplification.

\begin{definition}[Opt-Boundedness]\label{dfn:opt-bounded}
  A partitionable online problem is called \emph{opt-bounded}
  if there exists a constant $B$ such that
  $\forall s,s',x\colon |\cost{\algOpt_s(x)}-\cost{\algOpt_{s'}(x)}|\le B$.
\end{definition}

Note that the definition of opt-boundedness implies that any request sequence
$x$ is valid. In particular, the request sequence may end at any time.

\begin{definition}[Symmetric Problem]\label{dfn:symmetric}
  An online problem is called \emph{symmetric} if it is partitionable and every
  state is initial.
\end{definition}

Formally, any partitionable problem may be transformed into a symmetric one
simply by redefining the set of initial states.  However, this transformation
may significantly change the properties of the problem.
Now we are going to state the main result of this paper, namely that, under
certain conditions, the expected competitive ratio of symmetric problems can be
achieved w.h.p.

\begin{theorem}\label{thm:exptohp}
  Consider a opt-bounded symmetric online problem for which there is 
  a randomized online algorithm $\algA$ with constant expected competitive ratio $r$.
  Then, for any constant $\varepsilon>0$, there is a randomized online algorithm
  $\algA'$ with competitive ratio $(1+\varepsilon)r$ w.h.p. (with
  respect to the optimal cost).
\end{theorem}

We prove this theorem in the subsequent section. 

\section{Proof of Theorem~\ref{thm:exptohp}}\label{sec:mainthm}

For ease of presentation, we first provide a proof for a restricted setting
where the online problem at hand is also request-bounded.

The algorithm $\algA'$ simulates \algA and, on some specific places, performs
a \emph{reset} operation: if a part $x'$ of the input has been read so far,
and a corresponding output $y'$ has been produced, $(I,x',y')$
belongs to the same state as $(I',\emptyword,\emptyword)$, for some initial
configuration $I'$, because we are dealing
with a symmetric problem; hence, \algA can be restarted by $\algA'$ from $I'$. 

The general idea to boost the probability of acquiring a low cost is to
perform a reset each time the
algorithm incurs too much cost and to use Markov's inequality to bound the
probability of such an event.  However, the exact value of how much is ``too
much'' depends on the optimal cost of the input which is not known in advance.  Therefore,
the input is first partitioned into \emph{phases} of a fixed optimal cost, and
then each phase is cut into \emph{subphases} based on the cost incurred so far.
A reset may cause an additional expected cost of $r \cdot B$ for the subsequent phase
compared to an optimal strategy starting from another state, where $B$ is the constant of the 
opt-boundedness (Definition~\ref{dfn:opt-bounded}), \ie, $B$ bounds the different costs
between two optimal solutions for a fixed input for different states.
We therefore have to ensure that the phases are long enough so as to amortize this
overhead. 

From now on let us consider $\varepsilon$, $r$, $B$, and $F$ to be fixed
constants; recall that $F$ originates from the
request-boundedness property of the online problem at hand
(Definition~\ref{dfn:req-bounded}).  The algorithm $\algA'$ is parameterized by two
parameters $C$ and $D$ that depend on $\varepsilon$, $r$, $B$, and $F$.  These
parameters control the length of the phases and subphases, respectively, such
that $C+F$ delimits the optimal cost of one phase and $D+F$ delimits the cost
of the solution computed by $\algA'$ on one subphase; we require that $D > r(C
+ F + B)$.

Consider an input sequence $x=(x_1,\dots,x_n)$, an initial configuration
$I$, and let the optimal cost of the input $(I,x)$ be between $(k-1)C$ and
$kC$ for some integer $k$.  Then $x$ can be partitioned into $k$ phases
$\tilde{x}_1=(x_1,\dots,x_{n_2-1})$,
$\tilde{x}_2=(x_{n_2},\dots,x_{n_3-1}),\dots,
\tilde{x}_k=(x_{n_{k}},\dots,x_n)$ in such a way that $n_i$ is the minimal
index for which the optimal cost of the input $(I,(x_1,\dots,x_{n_i}))$ is at
least $(i-1)C$.  It follows that the optimal cost for one phase is at least
$C-F$ and at most $C+F$, with the exception of the last phase which may be
cheaper.  Note that this partition can be generated by the online algorithm
itself, \ie, $\algA'$ can determine when a next phase starts.  There are only
two reasons for $\algA'$ to perform a reset: at the beginning of each phase and
after incurring a cost exceeding $D$ since the last reset.  Hence, $\algA'$
starts each phase with a reset, and the processing of each phase is partitioned
into a number of subphases each of cost at least $D$ (with the exception of the
possibly cheaper last subphase) and at most $D+F$.

Now we are going to discuss the cost of $\algA'$ on a particular input.  Let us
fix the input $(I,x)$ which subsequently also fixes the indices
$1=n_1,n_2,\dots,n_k$.  Let $S_i$ be a random variable denoting the state of
the problem (according to Definition~\ref{dfn:state}) just before processing
request $x_i$, and let $W(i,j)$, $i\le j$, be a random variable denoting the
cost of $\algA'$ incurred on the input $x_i,\dots,x_j$. The following claim is
obvious.

\begin{claim}
  If $\algA'$ performs a reset just before processing $x_i$, then  
  $S_i$ captures all the information from the past 
  $W(i,j)$ depends on.
  In particular, if we fix $S_i=s$, $W(i,j)$ does not depend on $W(l_1,l_2)$,
  for any $l_1\le l_2\le i$ and any state $s$.
  
\end{claim}

The overall structure of the proof is as follows.
We first
show in Lemma~\ref{lm:Wexp} that the expected cost incurred during a phase (conditioned by the
state in which the phase was entered) is at most $\mu:=r(C + F + B)/(1-p)$,
where $p:=r(C + F + B)/D < 1$.
We can then consider variables $Z_0,Z_1,\dots,Z_k$ such that
\[ Z_0 := k\mu,\qquad  Z_{i} := (k-i)\mu + \sum_{j=1}^i \bW_j\quad \text{for }i>0, \]
where $\bW_i$ is the cost of the $i$th phase, clipped from above by some logarithmic bound, \ie,
\[ \bW_i := \min\{ W({n_i},{n_{i+1}}-1), c\log k\}, \]
for some suitable constant $c$.  We show in Lemma~\ref{lm:martingale} that  $Z_0,Z_1,\dots,Z_k$
form a bounded supermartingale, and then use the \conbound to conclude that
$Z_k$ is unlikely to be much larger than $Z_0$. By a suitable choice of the free parameters,
this implies that $Z_k$ is unlikely to be much larger than the expected cost of $\algA$.
Finally, we show that w.h.p. $Z_k$ is the cost of the algorithm $\algA'$.

In order to argue about the expected cost of a given phase in  Lemma~\ref{lm:Wexp},
let us first show that a phase is unlikely to have many subphases. For the rest of the 
proof, let $X_j$ be the random variable denoting the number of subphases of phase $j$.

\begin{lemma}\label{lm:Xprob}
  For any $i$, $s$, and any $\delta\in\nat$ it holds that 
  $\prob[X_{i}\ge\delta \mid  S_{n_i}=s]\le p^{\delta-1}$.
\end{lemma}

\begin{proof}
  The proof is done by induction on $\delta$.  For $\delta=1$ the statement holds by definition.
  Let $\ind_c$ denote the index of the first request after $c-1$ subphases, with
  $\ind_1={n_i}$, and $\ind_c=\infty$ if there are less than $c$ subphases.  In order
  to have at least $\delta\ge2$ subphases, the algorithm must enter some suffix
  of phase $i$ at position $\ind_{\delta-1}$ and incur a cost of more than $D$ (see
  Fig.~\ref{fig:subphases}).  Hence,
  \begin{align}
    \label{lm:Xprob:eq:1}
    \prob[X_{i}\ge\delta\mid  S_{n_i}=s] = {} & \prob[\ind_{\delta-1}<n_{i+1}-1\mid S_{n_i}=s]\\
    &\cdot \prob[ W(\ind_{\delta-1},n_{i+1}-1)>D \mid \ind_{\delta-1}<n_{i+1}-1\wedge S_{n_i}=s].\nonumber
  \end{align}
  \begin{figure}[ht]
    \centerline{\includegraphics[width=10cm]{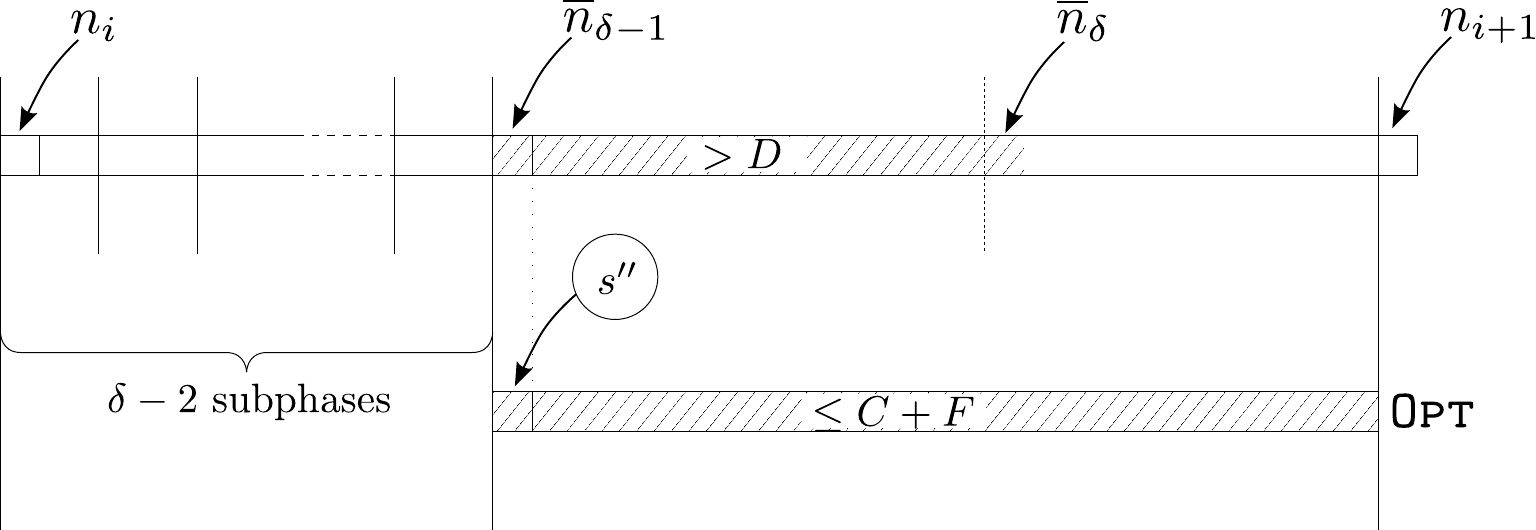}}
    \caption{The situation with $\delta$ subphases.}
    \label{fig:subphases}
  \end{figure}
  The fact that $\ind_{\delta-1}<n_{i+1}-1$ means that there are at least $\delta-1$ subphases, \ie,
  \begin{align}\label{lm:Xprob:eq:2}
    \prob[\ind_{\delta-1}<n_{i+1}-1\mid S_{n_i}=s] =  \prob[X_{i}\ge\delta-1\mid  S_{n_i}=s] \le p^{\delta-2}
  \end{align}
  by the induction hypothesis. Further, we can decompose
  \begin{align}
    \label{lm:Xprob:eq:3}
    &\prob[ W(\ind_{\delta-1}, n_{i+1}-1)>D \mid \ind_{\delta-1}<n_{i+1}-1\wedge S_{n_i}=s] \\
    &= \sum_{i',s'\atop n_i\le i'<n_{i+1}-1} 
      \prob[ W(\ind_{\delta-1}, n_{i+1}-1)>D \mid \ind_{\delta-1}=i'\wedge S_{i'}=s'\wedge S_{n_i}=s]\nonumber\\
    &\quad\cdot \prob[\ind_{\delta-1}=i'\wedge S_{i'}=s' \mid  \ind_{\delta-1}<n_{i+1}-1\wedge S_{n_i}=s].\nonumber
  \end{align}
  Now let us argue about the probability 
  \[ \prob[ W(\ind_{\delta-1}, n_{i+1}-1)>D \mid \ind_{\delta-1}=i'\wedge S_{i'}=s'\wedge S_{n_i}=s]. \]
  The algorithm $\algA'$ performed a reset just before reading $x_{i'}$, 
  so it starts simulating \algA from state $s'$. However, in the optimal solution,
  there is some state $s''$ associated with position $i'$ such that the cost of the remainder of
  the $i$th phase is at most $C+F$. Due to the assumption of the theorem, the optimal cost
  on the input $x_{i'},\dots,x_{n_{i+1}-1}$ starting from state $s'$ is at most $C+F+B$, and the 
  expected cost incurred by \algA is at most $r(C+F+B)$. Using Markov's inequality, we get
  \begin{equation}
    \label{lm:Xprob:eq:4}
    \prob[W(\ind_{\delta-1}, n_{i+1}-1)>D \mid \ind_{\delta-1}=i'\wedge S_{i'}=s'] \le \frac{r(C + F + B)}{D}=p. 
  \end{equation}
  Plugging~\eqref{lm:Xprob:eq:4} into~\eqref{lm:Xprob:eq:3}, and then together with~\eqref{lm:Xprob:eq:2}
  into~\eqref{lm:Xprob:eq:1} yields the result.
  
\end{proof}

Now we can argue about the expected cost of a phase.

\begin{lemma}\label{lm:Wexp}
  For any $i$ and $s$ it holds that 
  $\ev[W(n_i,n_{i+1}-1) \mid S_i=s]\le\mu$.
\end{lemma}

\begin{proof}  
  Let $\ind_c$ be defined as in the proof of Lemma~\ref{lm:Xprob}.
  Using the same arguments, we have that the expected cost of a single subphase is
  \[ \ev[ W(\ind_c,\min\{\ind_{c+1},n_{i+1}-1\}) \mid \ind_c=i'\wedge S_{i'}=s' ] \le r(C+F+B). \]
  Conditioning and decomposing by $\ind_c$ and $s'$, we get that
  \[ \ev[ W(\ind_c,\min\{\ind_{c+1},n_{i+1}-1\}) \mid X_i\ge c] \le r(C+F+B). \]
  Finally, let
  $Q_{i,c}=  W(\ind_c,\min\{\ind_{c+1},n_{i+1}-1\})$ if $X_i\ge c$, or $0$ if $X_i<c$. 
  This gives
  \begin{align*}
    &\ev[W(n_i, n_{i+1}-1)  \mid S_i=s]  = \sum_{c=1}^\infty \ev[ Q_{i,c} \mid S_i=s] \\
    & = \sum_{c=1}^\infty \ev[ Q_{i,c} \mid S_i=s \wedge X_i\ge c]\cdot \prob[X_i\ge c]\\
    &\le \sum_{c=1}^\infty r(C+F+B) p^{j-1} = r(C+F+B)/(1-p).
  \end{align*}
\end{proof}

Once the expected cost of a phase is established, we can construct the supermartingale as follows.

\begin{lemma}\label{lm:martingale}
  For any constant $c>0$, the sequence $Z_0,Z_1,\dots,Z_k$ is a supermartingale.
\end{lemma}

\begin{proof}
  Consider a fixed $c$.
  We have to show that for each $i$, $\ev[Z_{i+1}\mid Z_0,\dots,Z_i]\le Z_i$.
  From the definition of the $Z_i$'s it follows that $Z_{i+1}-Z_i=\bW_{i+1}-\mu$.
  Consider any elementary event $\xi$ from the probability space, and let
  $Z_i(\xi)=z_i$, for $i=0,\dots,k$ be the values of the corresponding random variables.
  We have
  \begin{align*}
    & \ev[Z_{i+1}\mid Z_0,\dots,Z_i](\xi) = \ev[ Z_{i+1}\mid Z_0=z_0,\dots,Z_i=z_i]\\
    &= \ev [ Z_i+\bW_{i+1}-\mu\mid Z_0=z_0,\dots,Z_i=z_i]
    = z_i-\mu+\ev[\bW_{i+1}\mid Z_0=z_0,\dots,Z_i=z_i]\\
    &\textstyle = z_i-\mu+\sum_s\ev[\bW_{i+1}\mid Z_0=z_0,\dots,Z_i=z_i,S_{n_{i+1}}=s]\\
    &\quad\cdot\prob[S_{n_{i+1}}=s \mid Z_0=z_0,\dots,Z_i=z_i]\\
    &\textstyle \le z_i-\mu+\sum_s\ev[W(n_{i+1},n_{i+2}-1)\mid S_{n_{i+1}}=s]\cdot\prob[S_{n_{i+1}}=s \mid Z_0=z_0,\dots,Z_i=z_i]\\
    &\textstyle \le z_i-\mu+\mu\sum_s\prob[S_{n_{i+1}}=s \mid Z_0=z_0,\dots,Z_i=z_i]=z_i=Z_i(\xi),
  \end{align*}
  where the last inequality is a consequence of Lemma~\ref{lm:Wexp}.
  
\end{proof}

Now we can use the following special case of the \conbound~\cite{Azu67,Hoe63}.

\begin{lemma}[Azuma, Hoeffding]\label{lm:Azuma}
  Let $Z_0,Z_1,\dots$ be a supermartingale, such that $|Z_{i+1}-Z_i|<\gamma$.
  Then for any positive real $t$, 
  \[ \prob[ Z_k - Z_0 \ge t]\le\exp\mathopen{}\left(-\frac{t^2}{2k\gamma^2}\right). \]
\end{lemma}

In order to apply Lemma~\ref{lm:Azuma}, we need the following bound.

\begin{claim}\label{clm:martbound}
  Let $k$ be such that $c\log k>\mu$.
  For any $i$ it holds that $|Z_{i+1}-Z_i|<c\log k$.
  
\end{claim}

We are now ready to prove the subsequent lemma.

\begin{lemma}\label{lm:mainbound}
  Let $k$ be such that $c\log k>\mu$. There is a constants $C$ (depending on $F$, $B$, $\varepsilon$, $r$)
  such that
  \[ \prob[Z_k\ge(1+\varepsilon)rkC]\le\exp\mathopen{}\left(-\frac{k\left((1+\varepsilon)rC-\mu\right)^2}{2c^2\log^2k}\right). \]
\end{lemma}

\begin{proof}
  Applying  Lemma~\ref{lm:Azuma} for any positive $t$, we get
  \[ \prob[Z_k-Z_0\ge t]\le\exp\mathopen{}\left(-\frac{t^2}{2kc^2\log^2k}\right). \]
  Noting that $Z_0=k\mu$, and choosing
  \[ t:=k((1+\varepsilon)rC-\mu) \]
  the statement follows. The only remaining task is to verify that $t>0$, \ie,
  that there is a constant $D$ such that
  \[ (1+\varepsilon)rC>r(C+F+B)\frac{1}{1-\frac{r(C+F+B)}{D}}. \]
  Let us choose $C$ such that $C>\frac{F+B}{\varepsilon}$. Then $(1+\varepsilon)C>C+F+B$, and it is
  possible to choose $D$ such that both $D>r(C+B+F)$ as required, and
  \[ D>\frac{(1+\varepsilon)r^2C(C+B+F)}{r((1+\varepsilon)C-(C+B+F))}. \]
  Thus, we have
  \[ rD(1+\varepsilon)C-rD(C+B+F)>(1+\varepsilon)r^2C(C+B+F) \]
  and therefore
  \[ (1+\varepsilon)rC(D-r(C+B+F))>rD(C+B+F) \]
  and the claim follows.
  
\end{proof}

To get to the statement of the main theorem, we show the following technical bound.

\begin{lemma}\label{lm:expprob}
  For any $c$, and $\beta>1$ there is a $k_0$ such that for any $k>k_0$
  \[ \exp\mathopen{}\left(-\frac{k\left((1+\varepsilon)rC-\mu\right)^2}{2c^2\log^2k}\right) \le \frac{1}{2(2+kC)^\beta}. \]
\end{lemma}

\begin{proof}
  Note that the left-hand side is of the form 
  $\exp\mathopen{}\left(-\eta\frac{k}{\log^2k}\right)$
  for some positive constant $\eta$. Clearly, for any $\beta>1$ and large enough $k$, it holds that
  $\exp\mathopen{}\left(\eta\frac{k}{\log^2k}\right)\ge2(2+kC)^\beta$.
  
\end{proof}

Combining Lemmata~\ref{lm:mainbound} and~\ref{lm:expprob}, we get the following result.

\begin{corollary}
  \label{crlr:probZk}
  There is a constant $C$ (depending on $F$, $B$, $\varepsilon$, $r$)
  such that for any $\beta>1$ there is a $k_0$ such that for any $k>k_0$ we have
  \[ \prob[Z_k\ge(1+\varepsilon)rkC]\le\frac{1}{2(2+kC)^\beta}. \]
\end{corollary}

In order to finish the proof of the main theorem we show that w.h.p., $Z_k$ is actually
the cost of the algorithm $\algA'$. 

\begin{lemma}
  \label{lm:Zcost}
  For any $\beta>1$ there is a $c$ and a $k_1$ such that for any $k>k_1$
  \[ \prob[Z_k\not=\cost{\algA'}]\le\frac{1}{2(2+kC)^\beta}. \]
\end{lemma}

\begin{proof}
  Since $Z_k=\sum_{j=1}^k\min\{W(n_j,n_{j+1}-1),c\log k\}$ the event that $Z_k\not=\cost{\algA'}$
  happens exactly when there exists some $j$ such that $W(n_j,n_{j+1}-1)>c\log k$.

  Consider any fixed $j$. Since the cost of a subphase is at most $D+F$, it holds that
  $W(n_j,n_{j+1}-1)\le X_j(F+D)$.  From 
  Lemma~\ref{lm:Xprob} it follows that for any $c$,
  \[ \prob[W(n_j,n_{j+1}-1)>c\log k]\le\prob\left[X_j\ge\left\lceil\frac{c\log k}{F+D}\right\rceil\right]
     \le p^{\frac{c\log k}{F+D}-1}. \]
  Consider the function
  \[ g(k):=\frac{\log\mathopen{}\left(\frac{2k}{p}(2+kC)^\beta\right)}{\log k}. \]
  It is decreasing, and $\lim_{k\mapsto\infty}g(k)=1+\beta$. Hence, it is possible to find a constant $c$,
  and a $k_1$ such that for any $k>k_1$ it holds that
  \[ c\ge\frac{F+D}{\log\mathopen{}\left(\frac{1}{p}\right)}\cdot g(k). \]
  From that it follows that
  \[ \frac{\log\mathopen{}\left(\frac{1}{p}\right)c\log k}{F+D}\ge
     \log\mathopen{}\left(\frac{1}{p}\right)+\log\mathopen{}\left(2k(2+kC)^\beta\right) \]
  and
  \[ \log\mathopen{}\left(\frac{1}{p}\right)\left(\frac{c\log k}{F+D}-1\right)\ge\log\mathopen{}\left(2k(2+kC)^\beta\right), \]
  \ie,
  \[ \left(\frac{1}{p}\right)^{\frac{c\log k}{F+D}-1}\ge 2k(2+kC)^\beta. \]
  Thus, for this choice of $c$ and $k_1$, it holds that
  \[ \prob[W(n_j,n_{j+1}-1)>c\log k]\le p^{\frac{c\log k}{F+D}-1}\le\frac{1}{2k(2+kC)^\beta}. \]
  Using the union bound, we conclude that the probability that the cost of any phase exceeds $c\log k$ is at most
  $1/(2(2+kC)^\beta)$.
  
\end{proof}

Using the union bound, combining Lemma~\ref{lm:Zcost} and Corollary~\ref{crlr:probZk}, and
noting that the cost of the optimum is at most $kC$, we get the following statement.

\begin{corollary}
  There is a constant $C$ such that 
  for any $\beta>1$ there a $k_2$ such that for any $k>k_2$ it holds
  \[ \prob[\cost{\algA'}\ge(1+\varepsilon)r\opt{}]\le\frac{1}{(2+kC)^\beta}.\]
\end{corollary}

To conclude the proof by showing that for any $\beta>1$ there is some $\alpha$ such that 
\[ \prob\left[\cost{\algA'}>(1+\varepsilon)r\opt{}
    + \alpha\right]\le\frac{1}{(2+kC)^\beta} \]
holds for all $k$, we have to choose $\alpha$ large enough to cover the cases
of $k<k_2$. For these cases, $\opt{}<k_2C$, and hence the
expected cost of \algA is at most $rk_2C$, and due to Lemma~\ref{lm:Wexp},
the expected cost of $\algA'$ is constant.  The right-hand side
$(2+kC)^{-\beta}$ is decreasing in $k$, so it is at least
$(2+k_2C)^{-\beta}$, which is again a constant.  From Markov's inequality
it follows that there exists a constant $\alpha$ such that 
\[ \prob\left[\cost{\algA'}>\alpha\right]<\frac{1}{(2+k_2C)^\beta} \]
finishing the proof of the restricted setting.

\subsection{Avoiding Request-Boundedness}

All that is left to do is to show how to handle problems that are not
request-bounded. The main idea is to apply the restricted Theorem~\ref{thm:exptohp} to a modified
request-bounded version of the given problem. We then have to show that there is
a modified version of the algorithm such that the computed solution has
an expected  competitive ratio close to the original one for the modified
problem. By ensuring that \emph{any} solution to the modified problem translates
to a solution of the original problem with at most the same competitive ratio,
it is enough to apply our theorem to the modified problem to obtain an analogous
result for the original problem.

Let $P$ be an opt-bounded symmetric problem; then $P$ is described entirely by
the feasible request-answer pairs (depending on the states), by its set of
states $\mathcal{S}$, and by costs of all request-answer pairs for all states.
Note that an expected $r$-competitive online algorithm \algA for $P$ has to have
an expected competitive ratio of $r$ for every request-answer pair.

Let $\cost[P]{s,\reqx,\ansy}$ denote the cost to give $\ansy$ as answer on
request $\reqx$ when in state $s$ of the problem $P$.
Let $\mathcal{Y}$ be the set of all possible
answers. Then we define the \emph{$(\sigma,\tau)$-truncated} version $P'$ of $P$ as follows.
Let $s$ be a state and $\reqx$ be a request; we set 
\[ m(s,\reqx) := \min_{\ansy \in \mathcal{Y}}\{\cost[P]{s,\reqx,\ansy}\}, \]
\ie, the minimal cost to answer $\reqx$ when in state $s$. In $P'$ we assign the cost
$ \cost[P']{s,\reqx,\ansy} =  \cost[P]{s,\reqx,\ansy}$, if $m(s,\reqx) \le \sigma$ and 
$ \cost[P']{s,\reqx,\ansy} =  \cost[P]{s,\reqx,\ansy} - m(s,\reqx) + \sigma$ otherwise.
We define all request-answer pairs of $P'$ such that $\cost[P']{s,\reqx,\ansy} > \tau$ to 
have a cost of $\infty$.
Both $P$ and $P'$ have the same remaining feasible
request-answer pairs for each state.

Note that any algorithm that gives an answer of cost $\infty$ with nonzero
probability cannot be competitive and that due to the modifications of the cost
function, some distinct states of $P$ may become a single state of $P'$. We will
abuse notation and ignore this fact because it does not change the proof. Thus
we assume that both problems have the same set of states.

We continue with some insights that help us to choose useful values for $\sigma$
and $\tau$.

\begin{claim}\label{claim:bounds}
  Given an expected $r$-competitive algorithm \algA for $P$, for any $\delta > 0$ there
  is a $(r+\delta)$-competitive online algorithm \algC for $P$ such that the cost
  $\cost[P]{s,\reqx,\ansy}$ for any $\ansy$ provided by \algC is at most
  $\delta^{-1}(\alpha + r \cdot (m(s,\reqx)+B))$. Furthermore, if $m(s,\reqx) \ge
  \frac{r}{r-1}B$, \algC may ignore the destination state and give a minimum cost
  answer greedily. 
\end{claim}

\begin{proof}
  Let $s'$ be the state selected after $s$ by an optimal solution and
  let $s''$ be the state when giving a greedy answer of cost $m(s,\reqx)$. 
  Let \optval{s}, \optval{s'}, and \optval{s''} be
  the costs of the respective optimal solutions when starting from $s$, $s'$, or
  $s''$.
   
  We first note that the optimal answer that leads from $s$ to $s'$ can
  have a cost of at most $m(s,\reqx)+B$ as otherwise, by the opt-boundedness,
  choosing greedily and moving to $s''$ would be a better solution.
  
  The sum of probabilities of \algA to select an answer of cost at least $\kappa$ is at
  most $(\alpha + r \cdot (m(s,\reqx)+B))/\kappa$ where
  the parameter $\alpha$ is due to the definition of the competitive ratio. 
  Otherwise the expected value
  would be too high if the adversary chooses to only send a single request. We
  set $\kappa = \delta^{-1}(\alpha + r \cdot (m(s,\reqx)+B))$
  to satisfy the $\delta$-closeness to the expected competitiveness. 
  We now show how to handle large values of $m(s,\reqx)$.
  To be $r$-competitive, we can afford a cost of
  \[
    r \cdot \optval{s} \ge r \cdot (m(s,\reqx) + \optval{s'}) \ge r \cdot
    (m(s,\reqx) + \optval{s''}-B).
  \]
  If we choose the first answer greedily and apply \algA for all remaining
  requests, the expected cost of the solution is at most
  \[
  m(s,\reqx) + r \cdot \optval{s''} \le m(s,\reqx) + r \cdot (\optval{s''}-B) + r\cdot B.
  \]
  Therefore, if $m(s,\reqx) \ge \frac{r}{r-1}B$, the modified solution is
  $r$-competitive.
  
\end{proof}

The claim suggests to set $\sigma = \frac{r}{r-1}B$ and 
$\tau = 2\varepsilon^{-1}(\alpha + r \cdot (\frac{r}{r-1}B + B))$, where we
chose $\delta = \varepsilon/2$. From now on $P'$ is the
$(\sigma,\tau)$-truncated version of $P$ with these values of $\sigma$ and
$\tau$.

As before, let \algA be an online algorithm for $P$ that computes a solution
with expected competitive ratio at most $r$. We design an algorithm \algB for $P'$ as follows. 
Suppose in state $s$ of $P'$, the adversary requests $\reqx$.
Then \algB simulates \algA in state $s$ on $\reqx$ within $P$. If $m(s,\reqx) \le \sigma$
and the answer $\ansy$ has a cost
smaller than $\tau$, the answer of $A'$ is $\ansy$. Otherwise \algB ignores the
answer of \algA and answers greedily while ignoring the destination state, and performing a reset subsequently.

It is clear that all answers of \algB are feasible for $P'$.
We first show that the expected competitive ratio of \algB for $P'$ is at most
$r + \varepsilon/2$. For each round with $m(s,x) \le \sigma$, the claim follows
directly from Claim~\ref{claim:bounds} using that any answer in $P$ with cost
higher than $\tau$ neither affects an optimal answer nor the algorithm's answer
due to the claim. Otherwise, if $m(s,x) > \sigma$, the competitive ratio of the
greedy answer is at most $r$, using the same argumentation as in the proof of
the second part of Claim~\ref{claim:bounds}.

To summarize, $P'$ is a symmetric, opt-bounded, and request-bounded problem and
\algB is an expected $(r+\varepsilon/2)$-competitive algorithm for $P'$.
Therefore, we can apply the restricted Theorem~\ref{thm:exptohp} as proven in
the last section with an error of $\varepsilon/2$ and with \algB to show that
there is an algorithm $\algA'$ that is $(r + \varepsilon)$-competitive for $P'$
w.h.p.

Finally we show that the competitive ratio in $P$ for any sequence of answers on
any request string cannot be larger than the competitive ratio of the same
sequence in $P'$.

Observe that a string of answers is optimal for $P$ if and only if it is optimal
for $P'$. Due to the opt-boundedness, an optimal solution cannot have any answer
on request $\reqx$ from state $s$ that has a cost larger than $m(s,x)+B$ in $P$ or
larger than $\sigma + B$ in $P'$. Therefore the parameter $\tau$ does not
influence any optimal solution in $P'$ and it cannot be an advantage to give an answer in
$P$ that is set to a cost of $\infty$ in $P'$. In each time step, the difference of the
cost of any answer $\ansy$ in $P$ and $P'$ given any state $s$ and request $\reqx$ is fixed to
exactly $m(s,\reqx) - \sigma$ as long as the answer has finite cost.  Thus, any improvement
of the answer sequence in one of the problems translates to an improvement in the other
one.

Let $z = z_1,z_2,\dotsc,z_k$ be an optimal sequence of answers and
$s'_1,s'_2,\dotsc,s'_k$ be the corresponding sequence of states.
Then it is
sufficient to show that for each $i$, the competitive ratio of $\algA'$ for $P$
is at most as high as the competitive ratio for $P'$. For any $i$, let us fix a state $s$ and
a request $\reqx$.
Let $\ansy$ be the answer given by $\algA'$. Then the
competitive ratio in $P'$ is $\cost[P']{s,\reqx,\ansy}/\cost[P']{s'_i,\reqx,z_i}$.
If $m(s,x) \le \sigma$, the cost of both the optimal answer and the algorithmic answer,
and therefore also the ratio, is identical in $P$ and $P'$. Otherwise, the ratio in
$P$ is 
\begin{align*}
& \cost[P]{s,\reqx,\ansy}/\cost[P]{s'_i,\reqx,z_i} \\
& = (\cost[P']{s,\reqx,\ansy} + m(s,x) - \sigma)/(\cost[P']{s'_i,\reqx,z_i} + m(s,x) - \sigma)\\
& \le \cost[P']{s,\reqx,\ansy}/\cost[P']{s'_i,\reqx,z_i},
\end{align*}
where the last inequality uses that any competitive ratio is at least one.

\section{Applications}\label{sec:applications}

We now discuss the impact of Theorem~\ref{thm:exptohp} on task systems, the
$k$-server problem, and paging.  Despite being related, these problems have
different flavors when analyzing them in the context of high probability
results.  Finally, we show that there are also problems that do not directly
fit into our framework but nevertheless allow for high probability results for
specific algorithms.

\subsection{Task Systems}\label{sec:task}

The properties of online problems needed for Theorem~\ref{thm:exptohp} are related
to the definition of task systems. There are, however, some important differences.

To analyze the relation, let us recall the definition of task systems as
introduced by Borodin et al.~\cite{BLS92}. We are given a finite state space
$S$ and a function $d\colon S \times S \rightarrow \real_+$ that specifies
the (finite) cost to move from one state to another. The requests given as input to a
task system are a sequence of $|S|$-vectors that specify, for each state, the
cost to process the current task if the system resides in that state. An online
algorithm for task systems aims to find a schedule such that the overall cost
for transitions and processing is minimized. From now on we will call states in
$S$ \emph{system states} to distinguish them from the states of
Definition~\ref{dfn:state}.  The main difference between states of
Definition~\ref{dfn:state} and system states is that states and the distances between states
depend on the requests provided as input and on the answers given by the online algorithm;
this way there may be infinitely many states. States are also more general than
system states in that we may forbid specific state transitions. 

\begin{theorem}\label{thm:tasksys}
  Let \algA be a randomized online algorithm with expected competitive ratio $r$ for task systems.
  Then, for any $\varepsilon > 0$, there is a randomized online algorithm $\algA'$ for task systems with 
  competitive ratio $(1+\varepsilon)r$ w.h.p. (with respect to the optimal cost).
\end{theorem}

\begin{proof}
  In a task system, the system states are exactly the states according to our definition, because
  the optimal future cost only depends on the current system state and a future request has the
  freedom to assign individual costs to each of the system states.
  In other words, an equivalence class $s$ from Definition~\ref{dfn:state} (\ie, one state) consists
  of exactly one unique system state.
  To apply Theorem~\ref{thm:exptohp}, we choose the constant $B$ of the theorem to be
  $\max_{s,t \in S}\{d(s,t)\}$. This way, the problem is opt-bounded as one
  transition of cost at most $B$ is sufficient to move to any system state used
  by an optimal computation.
  The problem is clearly partitionable according to
  Definition~\ref{dfn:partitionable} as each round is associated with a
  non-negative cost. The adversary may also stop after an arbitrary request.
  
  The remaining condition of Theorem~\ref{thm:exptohp} that every state is
  initial formally conflicts with the definition of task systems, because usually
  there is a unique initial configuration that corresponds to a state $s_0$.
  This problem is easy
  to circumvent by relabeling the states before each run (reset) of the algorithm, \ie,
  we construct an algorithm $\algA''$ that is used instead of $\algA$. When
  starting the computation, $\algA''$ determines the mapping and simulates the
  run of $\algA$ on the mapped instance. 
  Thus we are able to use Theorem~\ref{thm:exptohp} on $\algA''$ and the claim follows.
  
\end{proof}
\subsection{The \textit{k}-Server Problem}

The $k$-server problem, introduced by Manasse et al.~\cite{MMS90}, is
concerned with the movement of $k$ servers in a
metric space. Each request is a location and the algorithm has to move one of
the servers to that location. If the metric space is finite,
this problem is well known
to be a special metrical task system.  The states are all combinations of $k$
locations in the metric space and the distance between two states is the corresponding
minimum cost to move servers such that the new locations are reached. Each
request is a vector where all states but those containing the correct
destination have a processing time $\infty$ and the states containing the
destination have processing time zero. Using Theorem~\ref{thm:tasksys} this
directly implies that all algorithms with a constant expected competitive ratio
for the $k$-server problem in a finite metric space can be transformed into
algorithms that have almost the same competitive ratio w.h.p.

If the metric space is infinite, an analogous result is still valid except that
we have to bound the maximum transition cost by a constant.
This is the case, because the proof of Theorem~\ref{thm:tasksys} uses
the finiteness of the state space only to ensure bounded transition costs.

Without the restriction to bounded distances, in general we cannot obtain a
competitive ratio much better than the deterministic one w.h.p.

\begin{theorem}\label{thm:server}
  Let $(M,d)$ be a metric space with $|M|=n$ constant, $s \in M$ be the initial position of all
  servers, $\ell$ a constant and let $r$ be the infimum over the competitive ratios of all
  deterministic online 
  algorithms for the $k$-server problem in $(M,d)$ for instances with at most
  $\ell$ requests.
  For every $\varepsilon > 0$, there is a metric space $(M',d')$ where for any randomized online
  algorithm \algR for the $k$-server problem there is an oblivious adversary
  against which the solution of \algR has a competitive ratio of at least
  $r-\varepsilon$ with constant probability.
\end{theorem}

\begin{proof}
  We obtain $(M',d')$ as follows. The set $M'$ is composed of copies of
  $M\setminus \{s\}$. Let, for each $i \in \nat$, $M_i$ denote the $i$th copy of $M$ in $M'$
  together with the point $s$ (\ie, $s$ is in each of the sets $M_i$). This
  way $M = M_1$. For any pair of points $u,v \in M$ with copies $u_i,v_i$ in $M_i$, 
  we set $d'(u_i,v_i) = i \cdot d(u,v)$; we call $i$ the \emph{scaling factor} of $M_i$.
  For any $i \neq j$, the distance between points in
  distinct copies of $M$ is $d'(u_i,v_j) = d(s,u_i) + d(s,v_j)$.
  This way $(M',d')$ is a metric and we can choose freely a scaling factor for the cost
  function $d$.

  We now describe an adversary \adv that uses oblivious adversaries for
  deterministic online algorithms as black boxes and has two parameters $\lambda$ and $\zeta$ 
  that specify lower bounds on the number of requests and the cost of the
  optimal offline solution. \adv starts with
  $\lambda$ requests of the point $s$ in $M$ (\ie, the optimal cost after the
  first $\lambda$ requests is zero). Note that we cannot assume $\lambda$ to
  be a constant\punctuationfootnote{Without the first $\lambda$ requests, for a
  fixed online algorithm the only way to
  access more than constantly many random bits within $\ell$ requests is to
  use the random bits to decide on further access to the random tape. But then
  we could fix a constant probability to only access constantly many random
  bits. Thus, omitting $\lambda$ would strengthen the adversary and
  weaken this lower bound result more than it is acceptable.}.
  
  Afterwards the adversary starts a second phase where it simulates a
  deterministic adversary in a suitably scaled copy of $M$. 
  We assume without loss of generality that any considered
  algorithm is \emph{lazy}, \ie, it answers requests by only moving at most
  one server (see Manasse et al.~\cite{MMS90}).
  We choose as scaling factor  $j = \zeta \cdot \min_{u,v \in M}\{d(u,v)\}$.
  \adv sends all subsequent requests in $M_j$. 
  
  Due to the laziness assumption, after the first $\lambda$ requests there
  are at most $k^\ell$ different possibilities to answer the
  subsequent $\ell$ requests (we can view an answer simply as the index of one
  of the $k$ servers).  Adding also all shorter request sequences, by
  the geometric series there are at most 
  $\frac{k^{\ell+1}-1}{k-1} < k^{\ell+1}$ possible answer sequences. Analogously, there
  are less than $n^{\ell+1}$ possible request sequences of length at most
  $\ell$ in $M_j$. Thus, the total number of algorithms behaving differently within
  at most $\ell$ requests is less than $\psi=(k^{\ell+1})^{(n^{\ell+1})}$ and
  therefore constant.

  \adv may choose one of at most $\psi$ deterministic algorithms to play against.
  He analyzes the probability distribution of \algR's strategies after the first
  $\lambda$ requests.
  Then he selects one of the $\psi$ algorithms that corresponds to the strategy
  run by \algR with maximal probability. With \adv's choice of the algorithm, the
  competitive ratio of \algR is at least $r-\varepsilon$ with 
  constant probability at least $\psi^{-1}$ and the choice of $j$ ensures
  that the optimal cost is at least $\zeta$.
  
\end{proof}

\begin{corollary}
  If we allow the metric to be infinite, then there is no $(k-\varepsilon)$-competitive 
  online algorithm w.h.p. for the $k$-server problem for any constant $\varepsilon$.
\end{corollary}

We simply use that the lower bound of Manesse et al.~\cite{MMS90} satisfies the
properties of Theorem~\ref{thm:server}.

\subsection{Paging}\label{subsec:paging}

In the paging problem there
is a cache that can accommodate $k$ memory pages and the input consists of a
sequence of requests to memory pages.  If the requested page is in the cache, it
can be served immediately, otherwise some page must be evicted from the cache,
and be replaced by the requested page; this process is called a \emph{page fault}.
The aim of a paging algorithm is to generate as few page faults as possible.
Each request generates either cost $0$ (no page fault) or $1$ (page fault), and
the overall cost is the sum of the costs of the requests.
Paging can be seen as a $k$-server problem restricted to uniform metrics
where all distances are exactly one. In particular, the transition costs in that
metric are bounded.
Hence, the assumptions discussed in the previous subsection are fulfilled,
meaning that for any paging algorithm with expected competitive ratio $r$ there
is an algorithm with competitive ratio $r(1+\varepsilon)$ w.h.p.

Note that the marking algorithm is analyzed based on phases that correspond to
$k+1$ distinct requests, and hence the analysis of the expected competitive
ratio immediately gives the $2H_k-1$ competitive ratio also w.h.p. However,
\eg, the optimal algorithm with competitive ratio $H_k-1$ due to Achlioptas et al.~\cite{ACN00}
is a distribution-based algorithm where the high probability analysis is not
immediate; Theorem~\ref{thm:exptohp} gives an algorithm with competitive ratio
$H_k(1+\varepsilon)$ w.h.p. also in this case.

\subsection{Job Shop Scheduling with Unit Length Jobs}

In Section~\ref{sec:discussion} we will show that none of the conditions of
Theorem~\ref{thm:exptohp} can be omitted.  However, there are problems that do
not fit the assumptions of the theorem, and still can be solved almost
optimally by specific randomized online algorithms with high probability.
We use, however, a weaker notion of high probability than in the previous sections.

Consider the problem \emph{job shop scheduling with unit length tasks} (\jss
for short) defined as follows: We are given a constant number of $n$
\emph{jobs} $J_1$ to $J_n$ that consist of $m$ \emph{tasks} each.  Each such
task needs to be processed on a unique one of $m$ \emph{machines} which are
identified by their indices $1,2,\dots,m$, and we want to find a schedule with
the following properties.  Processing one task takes exactly $1$ time unit and,
since all jobs need every machine exactly once, we may represent them as
permutations $P_i=(p^i_1,p^i_2,\dots,p^i_m)$ of the machine indices, where
$p^i_j \in\{1,2,\dots,m\}$ for every $i \in \{1,2,\dots,n\}$ and
$j\in\{1,2,\dots,m\}$.  All $P_i$ arrive in an online fashion, that is, the
$(k+1)$th task of $P_i$ is not known before the $k$th task is processed.
Obviously, as long as all jobs request different machines, the work can be
parallelized.  If, however, at one time step, some of them ask for the same
machine, all but one of them have to be delayed.  The cost of a solution is
given by the total time needed for all jobs to finish all tasks; the goal is to
minimize this time (\ie, the overall makespan).

In the following, we use a graphical representation that was introduced by
Brucker~\cite{Bru88}. Let us first consider only two jobs $P_1$ and $P_2$.
Consider an $(m\times m)$-grid where we label the
$x$-axis with $P_1$ and the $y$-axis with $P_2$. The cell $(p^1_i,p^2_j)$ models
that, in the corresponding time step, $P_1$ processes a task on machine $p^1_i$
while $P_2$ processes a task on $p^2_j$.  A feasible schedule for the induced
instance of \jss is a path that starts at the upper-left vertex of the grid and
leads to the bottom right vertex.

\begin{figure}[t]
  \centerline{\includegraphics[width=6cm]{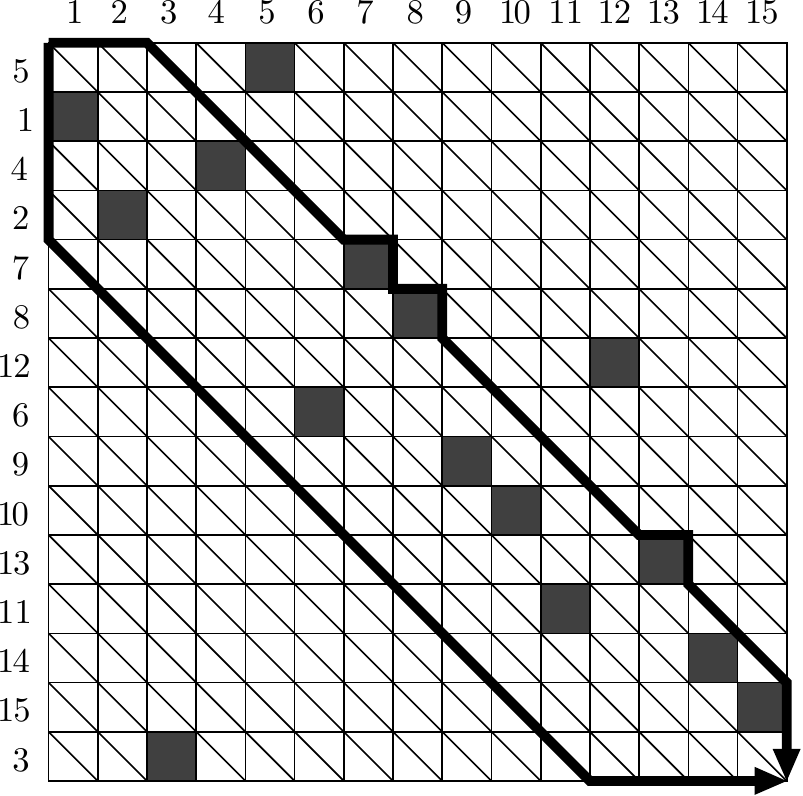}}
  \caption{An example input with two jobs each of size $15$
           and two strategies. Obstacles are marked by filled cells.}
  \label{fig:jssexample}
\end{figure}
It may use diagonal edges whenever $p^1_i\ne
p^2_j$.  However, if $p^1_i=p^2_j$, both $P_1$ and $P_2$ ask for the same machine at the
same time and therefore, one of them has to be delayed.  In this case, we say
that $P_1$ and $P_2$ collide and call the corresponding cells in the grid
\emph{obstacles} (see Fig.~\ref{fig:jssexample} for an
example with $m=15$).  If an algorithm has to delay a job, we say that it
\emph{hits an obstacle} and may therefore not make a diagonal move, but either
a horizontal or a vertical one. In the first case, $P_2$ gets delayed, in the
second case, $P_1$ gets delayed.  Note that, since $P_1$ and $P_2$ are
permutations, there is exactly one obstacle per row and exactly one obstacle
per column for every instance, therefore, $m$ obstacles overall for any
instance. The graphical representation generalizes naturally to the
$n$-dimensional case.

The problem has been studied previously, for instance in~\cite{BKKKM09,HMSW07,Hro05,KK11,Bru88}.
Hromkovi\v{c} et al.~\cite{HMSW07} showed the existence of a randomized online
algorithm \algR that achieves an expected competitive ratio of $1+2n/\sqrt{m}$,
for $n=o(\sqrt{m})$, assuming that it knows $m$.
\algR depends on diagonals in the grid; intuitively (in two or three dimensions), 
a diagonal in the grid is the sequence of integer points on a line that is parallel to the
line from the coordinate $(0,0,\dots,0)$ to $(m,m,\dots,m)$. 
More precisely, let $\mathcal{P}$ be the convex hull of the grid.
Then a diagonal is a sequence of integer points $d=\{d^i\}_1^k$ such that $d^1$ is in the
facet of $P$ that contains the origin $(0,0,\dots,0)$, $d^k$ is in the facet
containing the destination $(m,m,\dots,m)$, none of the two points is in a
smaller-dimensional face, and we obtain $d^{i+1}$ from
$d^{i}$ by increasing each coordinate by exactly one. As shown by Hromkovi\v{c}
et al.~\cite{HMSW07}, the number of diagonals that start at points with all coordinates
at most $r$ is exactly $n\cdot r^{n-1}$.

A diagonal template $D$ with respect to $r$ and $d$ is a sequence of consecutive
points in the grid that starts from $(0,0,\dots,0)$, moves to
$d^1$, visits each point of $d$ and finally moves to the destination
$(m,m,\dots,m)$. To reach $d^1$, $D$ delays each job $P_i$ by $r-d^1_i$ time
units in the begining and delays each job $P_i$ by $d^1_i$ time
units upon reaching the last point of the diagonal. Thus, a schedule that follows
a diagonal template without delays has a length of exactly $m+r$. A diagonal
strategy with respect to a diagonal template $D$ is a minimum-length schedule
that visits each point of $D$. Note that an online algorithm has all necessary
information to run a diagonal strategy, because when reaching an obstacle, all
possible ways to the subsequent point are available; an example of a diagonal
strategy is depicted in Fig.~\ref{fig:jssexample}.

The randomized algorithm \algR fixes the value $r$ and chooses
uniformly at random a diagonal $d$ with $\|d^1\|_\infty \le r$; then it follows
the corresponding diagonal strategy.

\begin{theorem}
  For any $r = o(\sqrt{m})$ there is an online algorithm for \jss that is $(1+f(m))$-competitive with
  probability $o_m(1)$, for any $f(m)=\omega(1/r)$.
\end{theorem}

\begin{proof}
  We already mentioned that \algR chooses one of $n\cdot r^{n-1}$ diagonals. It is
  also known that the total number of delays in all diagonal strategies caused by
  obstacles is at most $m\cdot{n \choose 2} \cdot (n-1) \cdot r^{n-2}$
  \cite{HMSW07}.
  Clearly, any schedule has a length of at least $m$.
  Thus, in order to be $(1+f(m))$-competitive, we need a diagonal
  strategy such that $(m+r+d)/m \le (1+f(m))$, where $d$ is the number of delays
  due to obstacles. Let $b$ be the number of
  diagonals considered by the algorithm such that the corresponding diagonal
  strategies have more than $d \le mf(m)-r$ delays caused by obstacles. Then, to
  show our claim, we have to ensure that $b/(n \cdot r^{n-1}) = o_m(1)$.
  
  The value of $b$ is maximized if we assume that any diagonal has either no
  obstacles or the delay is exactly $mf(m)-r$. Therefore,
  \[ b \le \frac{m\cdot{n \choose 2} \cdot (n-1) \cdot r^{n-2}}{mf(m)-r}. \]
  Since the dimension $n$ is a constant, the claim follows from
  \[ \frac{m\cdot{n \choose 2} \cdot (n-1) \cdot r^{n-2}}{(mf(m)-r)n \cdot r^{n-1}} 
     < \frac{m\cdot n^2}{(mf(m)-r) \cdot r} 
     = o_m(1). \]
\end{proof}
\section{Necessity of Requirements}\label{sec:discussion}

As mentioned above, our result holds with large generality as many well-studied
online problems meet the requirements we imposed.
However, the assumptions of Theorem~\ref{thm:exptohp} require that the problem
at hand
\begin{enumerate}
  \item is partitionable,
  \item every state is equivalent to some initial state, and \label{enum:init}
  \item $\forall s,s',x\colon
  |\cost{\algOpt_s(x)}-\cost{\algOpt_{s'}(x)}|\le B.$
  \label{enum:optb}
\end{enumerate}

As stated before, partitionability is not restrictive; every problem can be
presented as a partitionable one.  
We now show that removing any of the conditions
\ref{enum:init} and \ref{enum:optb} allows for a counterexample to the theorem.  
For the purpose of this discussion,
let $s$ and $s'$ in condition~\ref{enum:optb} range over all \emph{initial} states to
have it defined also for non-symmetric problems.

First, let us consider the following online problem where condition \ref{enum:init} is violated,
\ie, where not every state is equivalent to some initial state. 
There are
$n+2$  requests $x_{-1},x_0,\dots,x_n$. The request $x_{-1}=0$ is a dummy
request.  The request $x_0\in\{0,1\}$ is a test: if $y_{-1}=x_0$ the test is
\emph{passed}, otherwise the test is \emph{failed}; the cost of $y_{-1}$ and
$y_0$ is always zero.  For the remaining requests $x_1,\dots,x_n$ we have
$x_i\in\{0,1\}$. The cost of $y_i$, for $i=1,\dots,n$, is $1$ if the test has been
passed, or if $y_{i-1}=x_{i}$.  Otherwise, the cost of $y_i$ is $5$. The cost of
$y_n$ is zero.  The problem is clearly partitionable.  There are six states: the
initial state, then two possible states to guess the test, then one state for
processing all requests with the test passed, and two states for processing
requests with the test failed, based on the value of the previous answer.  From any
state, however, the optimal value of the remaining sequence of $m$ requests is
between $m$ and $m+6$.  A randomized online algorithm that guesses each time
independently has probability $1/2$ to pass the test incurring a cost of
$n$, and probability $1/2$ to fail, in which case, for any subsequent request,
it pays 1 with probability $1/2$, and $5$ with probability $1/2$. Putting everything 
together, the expected cost is $2n$, so $r=2$. On the other hand, for any
randomized algorithm, there is an input for which it has probability at least $1/2$
of failing the test, and then on each request probability at least $1/2$ of
a wrong guess. From  symmetry arguments we conclude that, once the test is
failed,  the probability that the algorithm makes at least $n/2-1$ wrong
guesses is at least $1/2$. Hence, with probability at least $1/4$ the cost of
the algorithm is at least $3n-4$, so it cannot be $c$-competitive w.h.p. for
any $c<3$.

Next, let us remove condition~\ref{enum:optb}.
We have seen a hint to the necessity in Theorem~\ref{thm:server}, but currently
no randomized online algorithm for the $k$-server problem is known to have a
competitive ratio better than $2k-1$ independent of the size of the metric
space. Therefore we give a second unconditional argument.
Let us consider the following problem: the states are pairs $(s,t)$ where
$s\in\{0,1\}$, $t\in\nat$, and any state can be an initial one.
Processing the request $r_i$ in state $(s,t)$ produces the answer
$y_i\in\{0,1\}$; the cost of $y_i$ is $2^t$ if $s=r_i$, and $3\cdot2^t$ if
$s\not=r_i$. After processing the request, the new state is $(y_i,t+1)$. It is
easy to verify that the problem is partitionable and that the states are in accord
with Definition~\ref{dfn:state}.  Also, it is easy to check that the worst-case
expected ratio of the algorithm that produces random answers is $2$.  On the
other hand, consider inputs that start from state $(0,0)$ with $x_1=0$. The
optimal cost is $2^n-1$, however, any randomized algorithm has probability at
least $1/4$ of incurring cost $9\cdot2^{n-2}$ (by failing the two last requests). 

\section{Conclusion}\label{sec:conclusion}

Our result opens several new questions. For instance, our results, so far, are
only shown for minimization problems.  Also note that our analysis does not
hold for the notion of \emph{strict} competitiveness (\ie, $\alpha=0$) for
arbitrary input sizes.  Furthermore, the assumption that all input strings are
feasible for all states (implied by the opt-boundedness) may allow for relaxations.

Until now, we only focused on upper bounds on the competitive ratio. Our
results, however, also open a potential lower bound technique: if a problem
satisfies our requirements, a lower bound w.h.p. implies a lower bound of
almost the same quality in expectation. In this context it is natural to ask
for the requirements of problems for a complementary result. How can we
determine the class of problems such that each algorithm that is
$r$-competitive w.h.p. can be transformed into an algorithm that is almost
$r$-competitive in expectation?

Finally, we would like to suggest the terminology to call a randomized online algorithm \algA
\emph{totally} $r$-competitive if, for any positive constant $\varepsilon$,
\algA is $c$-competitive in expectation and we may use
Theorem~\ref{thm:exptohp} to construct an online algorithm that
is $(r+\varepsilon)$-competitive w.h.p.
Analogously, an online problem is totally $c$-competitive if it admits a
totally $r$-competitive algorithm.

\subsection*{Acknowledgment}

The authors want to express their deepest thanks to Georg Schnitger who
gave some very important impulses that contributed to the results of this
paper.

\end{document}